\documentclass[11pt]{article}
\usepackage[top=1in,left=1in,right=1in,bottom=1in]{geometry}

\usepackage{wrapfig}
\usepackage{times}
\usepackage{fullpage}
\usepackage{color}
\usepackage[ruled,lined,linesnumbered,noend]{algorithm2e}
\usepackage{indentfirst}
\usepackage{graphicx}
\usepackage{booktabs}
\usepackage{epsf}
\usepackage{mdwlist}
\usepackage{mathrsfs}
\usepackage{amssymb}
\usepackage{amsthm}
\usepackage{bm}
\usepackage{array}
\usepackage{amsopn}
\usepackage{mathrsfs}
\usepackage{amsmath}
\usepackage{epsfig,subfigure,multirow}
\usepackage{fancybox,color}
\definecolor{ForestGreen}{rgb}{0.1333,0.5451,0.1333}
\usepackage{amsthm}

\newcommand{\myparagraph}[1]{\smallskip\noindent {\bf #1.}}
\SetArgSty{textrm}

\newcommand{\julian}[1]{{\bf Julian: #1}}
\newcommand{\yan}[1]{{\bf Yan: #1}}

\newcommand{\wcost}{\omega}
\newcommand{\memsize}{M}
\newcommand{\ourmodel}{$(\memsize,\wcost)$-ARAM}
\newcommand{\ourmodelfull}{$(\memsize,\wcost)$-Asymmetric RAM (ARAM)}
\newcommand{\fastmem}{small-memory}
\newcommand{\slowmem}{large-memory}
\newcommand{\work}{time}
\newcommand{\iocost}{ARAM cost}
\newcommand{\W}{T}
\newcommand{\boruvka}{Bor\r{u}vka}

\let \originalleft \left
\let\originalright\right
\renewcommand{\left}{\mathopen{}\mathclose\bgroup\originalleft}
\renewcommand{\right}{\aftergroup\egroup\originalright}

\newcommand{\mb}[1]{{\mbox{\emph{#1}}}}
\newcommand{\smb}[1]{{\mbox{\scriptsize\emph{#1}}}}
\newcommand{\SP}{\mbox{ShortestPath}}

\newtheorem{theorem}{Theorem}
\newtheorem{corollary}{Corollary}
\newtheorem{lemma}{Lemma}

\newcommand{\hide}[1]{}

\usepackage{microtype}

\title{Efficient Algorithms with Asymmetric Read and Write Costs}

\author{Guy
  E. Blelloch\\Carnegie Mellon University\\guyb@cs.cmu.edu \and
  Jeremy T. Fineman\\Georgetown
    University\\jfineman@cs.georgetown.edu \and Phillip
  B. Gibbons\\Carnegie Mellon University\\
    gibbons@cs.cmu.edu  \and Yan Gu\\Carnegie Mellon
    University\\ yan.gu@cs.cmu.edu \and Julian Shun\\UC Berkeley\\ jshun@eecs.berkeley.edu}

\begin{document}
    \maketitle
    \thispagestyle{empty}
    \par
\setcounter{secnumdepth}{2}
\setcounter{tocdepth}{2}

\begin{abstract}
In several emerging technologies for computer memory (main memory),
the cost of reading is significantly cheaper than the cost of writing.
Such asymmetry in memory costs poses a fundamentally different model
from the RAM for algorithm design.  In this paper we study lower and
upper bounds for various problems under such asymmetric read and write
costs.  We consider both the case in which all but $O(1)$ memory has
asymmetric cost, and the case of a small cache of symmetric memory.
We model both cases using the \ourmodel, in which there is a small
(symmetric) memory of size $\memsize$ and a large unbounded
(asymmetric) memory, both random access, and where reading from the
large memory has unit cost, but writing has cost $\wcost \gg 1$.

For FFT and sorting networks we show a lower bound cost of
$\Omega(\wcost n \log_{\wcost \memsize} n)$, which indicates that it
is not possible to achieve asymptotic improvements with cheaper reads
when $\wcost$ is bounded by a polynomial in $\memsize$.  Moreover,
there is an asymptotic gap (of $\min(\wcost,\log
n)/\log(\wcost \memsize)$) between the cost of sorting networks and
comparison sorting in the model.  This contrasts with the RAM,
and most other models, in which the asymptotic costs are the same.  We
also show a lower bound for computations on an $n \times n$ diamond
DAG of $\Omega(\wcost n^2/M)$ cost, which indicates no asymptotic
improvement is achievable with fast reads.  However, we show that for
the minimum edit distance problem (and related problems), which would
seem to be a diamond DAG, we can beat this lower bound with an
algorithm with only $O(\wcost n^2/ (M
\min(\wcost^{1/3},M^{1/2})))$ cost.  To achieve this we make use of a
``path sketch'' technique that is forbidden in a strict DAG
computation.  Finally, we show several interesting upper bounds for
shortest path problems, minimum spanning trees, and other problems.
A common theme in many of the upper bounds is that they require
redundant computation and a tradeoff between reads and writes.

\end{abstract}
\newpage
\setcounter{page}{1}
\section{Introduction}

Fifty years of algorithms research has focused on settings in which
reads and writes (to memory) have similar cost.  But what if reads and writes
to memory have significantly \emph{different} costs?  How would that impact
algorithm design?  What new techniques are useful for trading-off doing more
cheaper operations (say more reads) in order to do fewer expensive operations (say
fewer writes)?  What are the fundamental limitations on such trade-offs (lower bounds)?
What well-known equivalences for traditional memory fail to hold for
asymmetric memories?

Such questions are coming to the fore with the arrival of new
\textbf{main-memory} technologies~\cite{hp-nvm15,intel-nvm15} that
offer key potential benefits over existing technologies such as DRAM,
including nonvolatility, signicantly lower energy consumption, and
higher density (more bits stored per unit area).  These emerging
memories will sit on the processor's memory bus and be accessed at
byte granularity via loads and stores (like DRAM), and are projected
to become the dominant main memory within the
decade~\cite{Meena14,Yole13}.\footnote{While the exact technology is
  continually evolving, candidate technologies include phase-change
  memory, spin-torque transfer magnetic RAM, and memristor-based
  resistive RAM.}  Because these emerging technologies store data as
``states'' of a given material, the cost of reading (checking the
current state) is significantly cheaper than the cost of writing
(modifying the physical state of the material): Reads are up to an
order of magnitude or more lower energy, lower latency, and higher
(per-module) bandwidth than writes~\cite{Akel11,Athanassoulis12, BFGGS15,Carsonetal15,Dong09,Dong08,HuZXTGS14,ibm-pcm14b,Kim14,Qureshi12,Xu11}.

This paper provides a first step towards answering these fundamental
questions about asymmetric memories.  We introduce a simple model for studying such
memories, and a number of new results.  In the simplest
model we consider, there is an asymmetric random-access memory such
that reads cost 1 and writes cost $\wcost \gg 1$, as well as a constant
number of symmetric ``registers'' that can be read or written at unit
cost.  More generally, we consider settings in which the amount of
symmetric memory is $\memsize \ll n$, where $n$ is the input size:
We define the \emph{\ourmodelfull}, comprised of a symmetric
\fastmem{} of size $\memsize$ and an asymmetric \slowmem{} of unbounded
size with write cost $\wcost$.  The \emph{\iocost{} $Q$} is
the number of reads from \slowmem{} plus $\wcost$ times the number of
writes to \slowmem. The \emph{\work{} $\W$} is $Q$ plus the number of
reads and writes to \fastmem.

\newcommand{\mytag}{^\dag}

\begin{table}[t!]
\def\arraystretch{1.3}
\begin{center}
\caption{\label{tbl:results}
Summary of Our Results for the \ourmodel{} ($\mytag$indicates main results)}
\begin{tabular}{@{ }p{4.7cm}>{\centering}p{4.9cm}@{ }>{\centering}p{4.6cm}@{ }c@{ }}\toprule
\multirow{2}{*}{{\bf problem}} & {\bf \iocost{}} & {\bf \work{}} &
\multirow{2}{*}{{\bf section}} \\
& {\bf $Q(n)$ or $Q(n,m)$} & {\bf $\W(n)$ or $\W(n,m)$} & \\
\bottomrule
FFT & $\Theta(\wcost n \log n / \log(\wcost \memsize))\mytag$
  & $\Theta(Q(n) + n \log n)$ & \ref{sec:lowerbound}, \ref{sec:upper-intro} \\
\midrule
sorting networks & $\Omega(\wcost n \log n / \log(\wcost \memsize))\mytag$
  & $\Omega(Q(n) + n \log n)$ & \ref{sec:lowerbound} \\
sorting (comparison) & $O(n(\log n + \wcost))$
  & $\Theta(n(\log n + \wcost))$ & \ref{sec:upper-intro}, \cite{BFGGS15} \\
\hline
diamond DAG & $\Theta(n^2 \wcost/ \memsize)\mytag$
  & $\Theta(Q(n)+ n^2)$ & \ref{sec:lowerbound} \\
longest common subsequence, edit distance & \multirow{2}{*}{$O(n^2 \wcost/ \min(\wcost^{1/3}\memsize, \memsize^{3/2}))\mytag$}
  & $O(n^2 (1+\wcost / \min(\wcost^{1/3}\memsize^{2/3}, \memsize^{3/2})))\mytag$ & \multirow{2}{*}{\ref{sec:lcs}} \\
\hline
search tree, priority queue & $O(\wcost + \log n)$ per update
  & $O(\wcost + \log n)$ per update & \ref{sec:upper-intro} \\
\hline
2D convex hull, triangulation & $O(n(\log n + \wcost))$
  & $\Theta(n(\log n + \wcost))$ & \ref{sec:upper-intro} \\
\hline
BFS, DFS, topological sort, & \multirow{2}{*}{$\Theta(\wcost n + m)$}
  & \multirow{2}{*}{$\Theta(\wcost n + m)$} &
\multirow{2}{*}{\ref{sec:upper-intro}} \\
biconnected components, SCC & & & \\
\hline
\multirow{2}{*}{single-source shortest path} & $O(\min(n(\wcost+ m/\memsize), (m+n\log n)\wcost, m(\wcost+ \log n)))\mytag$
  & \multirow{2}{*}{$O(Q(n,m) + n \log n)$} & \multirow{2}{*}{\ref{sec:dijkstra}}\\
all-pairs shortest-path & $O(n^2(\wcost + n/\sqrt{\memsize}))$
  & $O(Q(n)+n^3)$ & \ref{sec:upper-intro} \\
\hline
minimum spanning tree & $O(m \min(\log n, n/\memsize) + \wcost n)\mytag$
  & $O(Q(n,m)+n\log n)$ & \ref{sec:MST} \\
\bottomrule
\end{tabular}
\end{center}
\vspace{-0.3in}
\end{table}

We present a number of lower and upper bounds for the \ourmodel, as
summarized in Table~\ref{tbl:results}.  These results consider a
number of fundamental problems and demonstrate how the asymptotic
algorithm costs decrease as a function of $M$, e.g., polynomially,
logarithmically, or not at all.

For FFT we show an $\Omega(\wcost n \log_{\wcost \memsize} n)$ lower
bound on the \iocost{}, and a matching upper bound.  Thus, even allowing
for redundant (re)computation of nodes (to save writes), it is not
possible to achieve asymptotic improvements with cheaper reads when
$\wcost \in O(\memsize^c)$ for a constant $c$.  Prior lower bound approaches for FFTs for
symmetric memory fail to carry over to asymmetric memory, so a new
lower bound technique is required.  We use an interesting new accounting
argument for fractionally assigning a unit weight for each node of the
network to subcomputations that each have cost $\wcost \memsize$.
The assignment shows that each subcomputation has on average at most
$\memsize \log (\wcost \memsize)$ weight assigned to it,
and hence the total cost across all $\Theta(n \log n)$ nodes yields the
lower bound.

For sorting, we show the surprising result that on asymmetric memories,
comparison sorting is asymptotically faster than sorting networks.
This contrasts with the RAM model (and I/O models, parallel models such as
the PRAM, etc.), in which the asymptotic costs are the same!
The lower bound leverages the same key partitioning lemma as in the FFT
proof.

We present a tight lower bound for DAG computation on
diamond DAGs that shows there is no asymptotic advantage of cheaper
reads.  On the other hand, we also show that allowing a vertex to be
``partially'' computed before all its immediate predecessors have been
computed (thereby violating a DAG computation rule), we can
beat the lower bound and show asymptotic advantage.
Specifically, for both the longest common subsequence and edit
distance problems (normally thought of as diamond DAG computations),
we devise a new ``path sketch'' technique that leverages partial
aggregation on the DAG vertices.  Again we know of no other models
in which such techniques are needed.

Finally, we show how to adapt Dijkstra's single-source shortest-paths
algorithm using phases so that the priority queue is kept in
\fastmem{}, and briefly sketch how to adapt \boruvka's minimum
spanning tree algorithm to reduce the number of shortcuts and hence
writes that are needed.  A common theme in many of our algorithms is
that they use redundant computations and require a tradeoff between
reads and writes.

\subparagraph*{Related Work} Prior work~\cite{BT06,Eppstein14,Gal05,
  nath:vldbj10,ParkS09,Viglas14} has studied read-write asymmetries in
NAND flash memory, but this work has focused on (i) the asymmetric
\emph{granularity} of reads and writes in NAND flash chips: bits can
only be cleared by incurring the overhead of erasing a large block of
memory, and/or (ii) the asymmetric \emph{endurance} of reads and
writes: individual cells wear out after tens of thousands of writes to
the cell.  Emerging memories, in contrast, can read and write
arbitrary bytes in-place and have many orders of magnitude higher
write endurance, enabling system software to readily balance
application writes across individual physical cells by adjusting its
virtual-to-physical mapping.  Other prior work has studied database
query processing under asymmetric read-write
costs~\cite{Chen11,Chen15,Viglas12,Viglas14} or looked at other systems
considerations~\cite{ChoL09,HuZXTGS14,
  LeeIMB09,yang:iscas07,ZhouZYZ09,ZWT13}.  Our recent
paper~\cite{BFGGS15} introduced the general study of parallel (and
external memory) models and algorithms with asymmetric read and write
costs, focusing on sorting.  Our follow-on paper~\cite{BBFGGMS16}
defined an abstract nested-parallel model of computation with
asymmetric read-write costs that maps efficiently onto more concrete
parallel machine models using a work-stealing scheduler, and presented
reduced-write, work-efficient, highly-parallel algorithms for a number
of fundamental problems such as tree contraction and convex hull.  In
contrast, this paper considers a much simpler model (the sequential
\ourmodel) and presents not just algorithms but also lower bounds---plus,
the techniques are new.  Finally, concurrent with this
paper, Carson et al.~\cite{Carsonetal15} developed interesting upper
and lower bounds for various linear algebra problems and direct N-body
methods under asymmetric read and write costs.  For sequential
algorithms, they define a model similar to the \ourmodel{} (as well as a
cache-oblivious variant), and show that for ``bounded data reuse''
algorithms, i.e., algorithms in which each input or computed value is
used only a constant number of times, the number of writes to
asymmetric memory is asymptotically the same as the sum of the reads
and writes to asymmetric memory.  This implies, for example, a tight
$\Omega(n \log n/\log M)$ lower bound on the number of writes for FFT
under the bounded data reuse restriction; in contrast, our tight
bounds for FFT do not have this restriction and use fewer writes.
They also presented algorithms without this restriction for matrix
multiplication, triangular solve, and Cholesky factorization that
reduce the number of writes to $\Theta(\mbox{output size})$,
without increasing the
number of reads, as well as various distributed-memory parallel algorithms.

\section{Model and Preliminaries}
\label{sec:prelim}

We analyze algorithms in an \emph{\ourmodel}.  In the model we assume
a symmetric \emph{\fastmem} of size $\memsize \geq 1$, an asymmetric
\emph{\slowmem} of unbounded size, and a \emph{write cost}
$\wcost \geq 1$, which we assume without loss of generality is an integer.
(Typically, we are interested in the setting where $n \gg \memsize$,
where $n$ is the input size, and $\wcost \gg 1$.)
We assume standard random access machine (RAM) instructions.  We
consider two cost measures for computations in the model.  We define
the (asymmetric) \emph{\iocost} $Q$ as the total number of reads from
\slowmem{} plus $\wcost$ times the number of writes to \slowmem.  We define the
(asymmetric) \emph{\work} $\W$ as the \iocost{} plus the number of reads
from and writes to \fastmem.\footnote{The \work{} metric models the fact
  that reads to certain emerging asymmetric memories
  are projected to be roughly as fast as reads to symmetric memory
  (DRAM). The \iocost{} metric $Q$ does not make this assumption and hence is more generally applicable.}  Because
all instructions are from memory, this includes any cost of computation.  In
the paper we present results for both cost measures.

The model contrasts with the widely-studied external-memory
model~\cite{AggarwalV88} in the asymmetry of the read and write costs.
Also for simplicity in this paper we do not partition the memory into
blocks of size $B$.  Another difference is that the asymmetry implies
that even the case of $\memsize=O(1)$ (studied in~\cite{BFGGS15} for
sorting) is interesting.  We note that our \iocost{}
is a special case of the general flash model cost proposed
in~\cite{Ajwani2009}; however that paper presents algorithms only for
another special case of the model with symmetric read-write costs.


We use the term \emph{value} to refer to an object that fits in one
word (location) of the memory.  We assume words are of size
$\Theta(\log n)$ for input size $n$.  The size $\memsize$ is the
number of words in \fastmem.  All logarithms are base 2 unless
otherwise noted.
The DAG computation problem is given a DAG and a value for each of its
input vertices (in-degree = 0), compute the value for each of its
output vertices (out-degree = 0).  The value of any non-input vertex
can be computed in unit time given the value of all its immediate
predecessors.  As in standard I/O models~\cite{AggarwalV88} we assume
values are atomic and cannot be split when mapped into the memory.
The DAG computation problem can be modeled as a pebbling game on the
DAG~\cite{HPV77}.  Note that we allow (unbounded) recomputation of a
DAG vertex, and indeed recomputation is a useful technique for
reducing the number of writes (at the cost of additional reads).

\newcommand{\icount}{l}
\newcommand{\ocount}{m}

\section{Lower Bounds}
\label{sec:lowerbound}

We start by showing lower bounds for FFT DAGs, sorting networks and
diamond DAGs.  The idea in showing the lower bounds is to
partition a computation into subcomputations that each have a lower bound on
cost, but an upper bound on the number of inputs and outputs
they can use.  Our lower bound for FFT DAGs then uses an interesting
accounting technique that gives every node in the DAG a unit weight,
and fractionally assigns this weight across the subcomputations.  In
the special case $\wcost = 1$, this leads to a simpler proof for
the lower bound on the I/O complexity of FFT DAGs than the well-known
bound by Hong and Kung~\cite{HK81}.

We refer to a \emph{subcomputation} as any contiguous sequence of
instructions.  The \emph{outputs} of a subcomputation are the values
written by the subcomputation that are either an output of the full
computation or read by a later subcomputation.  Symmetrically, the
\emph{inputs} of a subcomputation are the values read by the
subcomputation that are either an input of the full computation or
written by a previous subcomputation.  The \emph{space} of a
computation or subcomputation is the number of memory locations both
read and written.  An $(\icount,\ocount)$-\emph{partitioning} of a
computation is a partitioning of instructions into subcomputations
such that each has at most $\icount$ inputs and at most $\ocount$
outputs.
We allow for recomputation---instructions in different subcomputations
might compute the same value.

\begin{lemma}
\label{lem:partitions}
Any computation in the \ourmodel{} 
has an $((\wcost + 1)\memsize,2\memsize)$-partitioning such that at
most one of the subcomputations has \iocost{} $Q < \wcost \memsize$.
\end{lemma}
\begin{proof}
  We generate the partitioning constructively.  Starting at the
  beginning, partition the instructions into contiguous blocks such that
  all but possibly the last block has cost $Q \geq \wcost \memsize$,
  but removing the last instruction from the block would have cost $Q <
  \wcost \memsize$.  To remain within the cost bound each such
  subcomputation can read at most $\wcost \memsize$ values from
  \slowmem.  It can also read the at most $\memsize$ values that are
  in the \fastmem{} when the subcomputation starts.  Therefore it can
  read at most $(\wcost + 1)\memsize$ distinct values from the input
  or from previous subcomputations.  Similarly, each subcomputation
  can write at most $\memsize$ values to \slowmem, and an additional
  $\memsize$ that remain in \fastmem{} when the subcomputation ends.
  Therefore it can write at most $2\memsize$ distinct values that are
  available to later subcomputations or the output.
\end{proof}

\myparagraph{FFT} We now consider lower bounds for the DAG computation
problem for the family of FFT DAGs (also called FFT networks, or
butterfly networks).  The FFT DAG of input size $n = 2^k$ consists of
$k + 1$ levels each with $n$ vertices (total of $n \log 2 n$
vertices).  Each vertex $(i,j)$ at level $i \in {0,\ldots,k-1}$ and
row $j$ has two out edges, which go to vertices $(i+1,j)$ and $(i+1,j
\oplus 2^i)$ ($\oplus$ is the exclusive-or of the bit representation).
This is the DAG used by the standard FFT (Fast Fourier Transform)
computation.  We note that in the FFT DAG there is at most a
single path from any vertex to another.

\begin{lemma}
\label{lem:fftpartitions}
Any $(\icount,\ocount)$-partitioning of a computation for simulating
an $n$ input FFT DAG has at least $n \log n / (\ocount \log \icount)$
subcomputations.
\end{lemma}
\begin{proof}
We refer to all vertices whose values are outputs of any
subcomputation, as \emph{partition output vertices}.  We assign each
such vertex arbitrarily to one of the subcomputations for which it
is an output.

Consider the following accounting scheme for fractionally assigning a
unit weight for each non-input vertex to some set of partition output
vertices.  If a vertex is a partition output vertex, then assign the
weight to itself.  Otherwise take the weight, divide it evenly between
its two immediate descendants (out edges) in the FFT DAG, and
recursively assign that weight to each.  For example, for a vertex $x$
that is not a partition output vertex, if an immediate descendant $y$
is a partition output vertex, then $y$ gets a weight of $1/2$ from
$x$, but if not and one of $y$'s immediate descendants $z$ is, then
$z$ gets a weight of $1/4$ from $x$.  Since each non-input vertex is
fully assigned across some partition output vertices, the sum of the
weights assigned across the partition output vertices exactly equals
$|V|-n = n \log n$.  We now argue that every partition output vertex
can have at most $\log \icount$ weight assigned to it.  Looking back
from an output vertex we see a binary tree rooted at the output.  If
we follow each branch of the tree until we reach an input for the
subcomputation, we get a tree with at most $\icount$ leaves, since
there are at most $\icount$ inputs and at most a single path from
every vertex to the output.  The contribution of each vertex in the
tree to the output is $1/2^i$, where $i$ is its depth (the root is
depth $0$).  The leaves (subcomputation inputs) are not included since
they are partition output vertices themselves, or inputs to the whole
computation, which we have excluded.  By induction on the tree
structure, the weight of that tree is maximized when it is perfectly
balanced, which gives a total weight of $\log \icount$.

Therefore since every subcomputation can have at most $\ocount$
outputs, the total weight assigned to each subcomputation is at most
$\ocount \log \icount$.  Since the total weight across all
subcomputations is $n \log n$, the total number of subcomputations is
at least $n \log n / (\ocount \log \icount)$.
\end{proof}

\begin{theorem}[FFT Lower Bound]
  Any solution to the DAG computation problem on the family of FFT
  DAGs parametrized by input size $n$ has costs $\displaystyle
  Q(n) = \Omega\left({\wcost n \log n}/{\log (\wcost
      \memsize)}\right)$ and $\W(n) = \Omega(Q(n) + n \log n)$ on the
  \ourmodel.
\end{theorem}
\begin{proof}
  By Lemma~\ref{lem:partitions} every computation must have an
  $((\wcost + 1)\memsize,2\memsize)$-partitioning with subcomputation
  cost $Q \geq \wcost \memsize$ (except perhaps one).  Plugging in
  Lemma~\ref{lem:fftpartitions} we have $Q(n) \geq \wcost \memsize n
  \log n / (2\memsize \log ((\wcost + 1)\memsize))$, which gives our
  bound on $Q(n)$.  For $\W(n)$ we just add in the cost of the
  computation of each vertex.
\end{proof}

Note that whichever of $\wcost$ and $\memsize$ is larger will dominate
in the denominator of $Q(n)$.  When $\wcost \leq \memsize$, these
lower bounds match those for the standard external memory
model~\cite{HK81,AggarwalV88} assuming both reads and writes have cost
$\wcost$.  This implies that cheaper reads do not help asymptotically
in this case.  When $\wcost > \memsize$, however, there is a potential
asymptotic advantage for the cheaper reads.

\myparagraph{Sorting Networks} A sorting network is a acyclic network
of comparators, each of which takes two input keys and returns the
minimum of the keys on one output, and the maximum on the other.  For
a family of sorting networks parametrized by $n$, each network takes
$n$ inputs, has $n$ ordered outputs, and when propagating the inputs
to the outputs must place the keys in sorted order on the outputs.  A
sorting network can be modeled as a DAG in the obvious way.  Ajtai,
Koml\'{o}s and Szemer{\'e}di~\cite{AKS83} described a family of
sorting networks that have size $O(n \log n)$ and depth $O(\log n)$.
Their algorithm is complicated and the constants are very large.  Many
simplifications and constant factor improvements have been made,
including the well known Patterson variant~\cite{Paterson90} and a
simplification by Seiferas~\cite{Seiferas09}.  Recently
Goodrich~\cite{Goodrich14} gave a much simpler construction of an $O(n
\log n)$ size network, but it requires polynomial depth.  Here we show
lower bounds of simulating any sorting network on the \ourmodel.

\begin{theorem}[Sorting Lower Bound]
\label{lem:sortinglower}
Simulating any family of sorting networks parametrized on input size
$n$ has $\displaystyle Q(n) = \Omega\left(\frac{\wcost n \log n}{\log
    (\wcost \memsize)}\right)$ and $\W(n) = \Omega(Q(n) + n \log n)$ on
the \ourmodel.
\end{theorem}
\begin{proof}
  Consider an $(\icount,\ocount)$-partitioning of the computation.
  Each subcomputation has at most $\icount$ inputs from the network,
  and $\ocount$ outputs for the network.  The computation is oblivious
  to the values in the network (it can only place the min and max on
the outputs of each comparator).  Therefore locations of the inputs and outputs are fixed
  independent of input values.  The total number of choices the
  subcomputation has is therefore $\binom{\icount}{\ocount} \ocount! =
  \icount! / (\icount-\ocount)! < \icount^{\ocount}$.  Since there are
  $n!$ possible permutations, we have that the number of
  subcomputations $k$ must satisfy $(\icount^{\ocount})^k \geq n!$.
  Taking logs of both sides, rearranging, and using Stirling's formula
  we have $k > \log(n!)/(\ocount \log \icount) > \frac{1}{2} n \log
  n/(\ocount \log \icount)$ (for $n > e^2$).  By
  Lemma~\ref{lem:partitions} we have $Q(n) > \wcost \memsize
  \frac{1}{2} n \log n /(2\memsize \log ((1 + \wcost) \memsize)) =
  \frac{1}{4}\wcost n \log n/\log ((1 + \wcost)\memsize)$ (for $n >
  e^2$).
\end{proof}

These bounds are the same as for simulating an FFT DAG, and, as with
FFTs, they indicate that faster reads do not asymptotically affect the lower bound unless $\wcost >
\memsize$.  These lower bounds rely on the sort being done on a
network, and in particular that the location of all read and writes
are oblivious to the data itself.  As discussed in the next section,
for general comparison sorting algorithms, we can get better upper
bounds than indicated by these lower bounds.

\myparagraph{Diamond DAG}
We consider the family of \emph{diamond DAGs} parametrized on size
$n$.  Each DAG has $n^2$ vertices arranged in a $n\times n$ grid such
that every vertex $(i,j), 0 \leq i < (n-1) ,0 \leq j < (n-1)$ has two
out-edges to $(i+1,j)$ and $(i,j+1)$.  The DAG has one input at
$(0,0)$ and one output at $(n-1,n-1)$.
Diamond DAGs have many applications in dynamic programs, such as for
the edit distance (ED), longest common subsequence (LCS), and optimal
sequence alignment problems.

\begin{lemma}[Cook and Sethi, 1976]
\label{lem:pspace}
Solving the DAG computation problem on the family of diamond DAGs of
input parameter $n$ (size $n \times n$) requires $n$ space to store
vertex values from the DAG.
\end{lemma}
\begin{proof}
  Cook and Sethi~\cite{CS76} show that evaluating the top half of a
  diamond DAG ($i + j \geq n-1$) , which they call a pyramid DAG,
  requires $n$ space to store partial results.  Since all paths of the
  diamond DAG must go through the top half, it follows for the diamond
  DAG.
\end{proof}

\begin{theorem}[Diamond DAG Lower Bound]
\label{lem:diamond}
The family of diamond DAGs parametrized on input size $n$ has
$\displaystyle Q(n) = \Omega\left(\frac{\wcost n^2}{\memsize}\right)$
and $\W(n) = \Omega(Q(n) + n^2)$ on the \ourmodel.
\end{theorem}
\begin{proof}
  Consider the sub-DAG induced by a $2\memsize \times 2\memsize$
  diamond ($a \leq i < a + 2\memsize, b \leq j < b + 2\memsize$) of
  vertices.  By Lemma~\ref{lem:pspace} any subcomputation that
  computes the last output vertex of the sub-DAG requires $2\memsize$
  memory to store values from the diamond.  The extra in-edges along
  two sides and out-edges along the other two can only make the
  problem harder.  Half of the $2\memsize$ required memory can be from
  \fastmem, so the remaining $\memsize$ must require writing those
  values to \slowmem.  Therefore every $2\memsize \times 2\memsize$
  diamond requires $\memsize$ writes of values within the diamond.
  Partitioning the full diamond DAG into $2\memsize \times 2\memsize$
  sub-diamonds, gives us $n^2/(2\memsize)^2$ partitions.  Therefore
  the total number of writes is at least $\displaystyle \memsize
  \times \frac{n^2}{(2\memsize)^2} = \frac{n^2}{4\memsize}$, each with
  cost $\wcost$.  For the \work{} we need to add the $n^2$ calculations
  for all vertex values.
\end{proof}

\noindent
This lower bound is asymptotically tight since a diamond DAG can be
evaluated with matching upper bounds by evaluating each $M/2 \times
M/2$ diamond sub-DAG as a subcomputation with $M$ inputs, outputs and
memory.

These bounds show that for the DAG computation problem on the family
of diamond DAGs there is no asymptotic advantage of having cheaper
reads.  In Section~\ref{sec:lcs} we show that for the ED and LCS
problems (normally thought of as a diamond DAG computation), it is
possible to do better than the lower bounds.  This requires breaking
the DAG computation rule by partially computing the values of each
vertex before all inputs are ready.  The lower bounds are interesting
since they show that improving asymptotic performance with cheaper
reads requires breaking the DAG computation rules.

\section{Upper Bounds}\label{sec:upper-intro}

We start in this section by showing that a variety of problems have
reasonably easy optimal upper bounds.  In the two sections that follow
and Appendix~\ref{sec:MST} we study problems that are more challenging.

\subparagraph*{FFT}
For the FFT we can match the lower bound using the algorithm described
elsewhere~\cite{BFGGS15}, although in that case the computation cost
was not considered.  The idea is to first split the DAG into layers of
$\log(\wcost \memsize)$ levels.  Then divide each layer so that the
last $\log \memsize$ levels are partitioned into FFT networks of
output size $\memsize$.  Attach to each partition all needed inputs
from the layer and the vertices needed to reach them (note that these
vertices will overlap among partitions).  Each extended partition will
have $\wcost \memsize$ inputs and $\memsize$ outputs, and can be
computed in $\memsize$ \fastmem{} with $Q = O(\wcost \memsize)$, and
$\W = O((\wcost + \log \memsize) \memsize)$.  This gives a total upper
bound of $Q = O(\wcost \memsize \times n \log n /(\memsize \log(\wcost
\memsize))) = O(\wcost n \log n /\log (\wcost \memsize))$, and $\W =
O(Q(n) + n \log n)$, which matches the lower bound (asymptotically).
All computations are done within the DAG model.

\subparagraph*{Search Trees and Priority Queues}
We now consider algorithms for some problems that can
be implemented efficiently using
balanced binary search trees.  In the following
discussion we assume $\memsize = O(1)$.  Red-black trees with
appropriate rebalancing rules require only $O(1)$ amortized \work{}
per update (insertion or deletion) once the location for the key is
found~\cite{Tarjan83}.  For a tree of size $n$ finding a key's
location uses $O(\log n)$ reads but no writes, so the total amortized
cost $Q = \W = O(\wcost + \log n)$ per update in the \ourmodel.
For arbitrary sequences of searches and updates, $\Omega(\wcost + \log n)$
is a matching lower bound on the amortized cost per operation
when $\memsize = O(1)$.
Because priority queues can be implemented with a binary search tree,
insertion and delete-min have the same bounds.  It seems more
difficult, however, to reduce the number of writes for priority queues
that support efficient melding or decrease-key.

\subparagraph*{Sorting}
Sorting can be implemented with $Q = \W = O(n (\log n + \wcost))$ by
inserting all keys into a red-black tree and then reading them off in
priority order~\cite{BFGGS15}.  We note that this bound on \work{} is
better than the sorting network lower bound
(Theorem~\ref{lem:sortinglower}). For example, when $\wcost = \memsize
= \log n$ it gives a factor of $\log n / \log \log n$ improvement.
The additional power is a consequence of being able to randomly write
to one of $n$ locations at the leaves of the tree for each insertion.
The bound is optimal for $\W$ because $n$ writes are required for the
output and comparison-based sorting requires $O(n \log n)$ operations.

\subparagraph*{Convex Hull and Triangulation}
A variety of problems in computational geometry can be solved
optimally using balanced trees and sorting.  The planar convex-hull
problem can be solved by first sorting the points by $x$ coordinates
and then either using Overmars' technique or Graham's
scan~\cite{DCKO08}.  In both cases, the second part takes linear \work{}
so the overall cost is $O(\mb{Sort}(n))$.  The planar Delaunay
triangulation problem can be solved efficiently with the plane sweep
method~\cite{DCKO08}.  This involves maintaining a priority queue on
$x$ coordinate, and maintaining a balanced binary search tree on the
$y$ coordinate.  A total of $O(n)$ operations are required on each,
again giving bounds $O(\mb{Sort}(n))$.

\subparagraph*{BFS and DFS} Breadth-first and depth-first search can be
performed with $Q = \W = O(\wcost n + m)$.  In particular each vertex
only requires a constant number of writes when it is first added to
the frontier (the stack or queue) and a constant number of writes when
it is finished (removed from the stack or queue).  Searches along an
edge to an already visited vertex require no writes.  This implies
that several problems based on BFS and DFS also only require $Q = \W =
O(\mb{DFS}(n))$.  Such problems include topological sort, biconnected
components, and strongly connected components.  The analysis is based
on the fact that there are only $O(n)$ forward edges in the DFS.
However when using priority-first search on a weighted graph (e.g.,
Dijkstra's or Prim's algorithm) then the problem is more difficult to
perform optimally as the priority queue might need to
be updated for every visited edge.

\subparagraph*{Dynamic Programming}
With regards to dynamic programming, some problems are reasonably easy
and some harder.
We covered LCS and ED in Section~\ref{sec:lcs}.
The standard Floyd-Warshall algorithm for the
all-pairs shortest-path (APSP) problem uses $O(n^3)$ writes.  However,
by rearranging the loops and carefully scheduling the writes it is
possible to implement the algorithm using only $O(n^2)$
writes and $O(n^3)$ reads, giving
$\W = O(\wcost n^2 + n^3)$ (Appendix~\ref{sec:fw}).
This version, however, is not
efficient in terms of $Q$.  Kleene's divide-and-conquer
algorithm~\cite{Aho74} can be used to reduce the \iocost{}~\cite{PPP04}.
Each recursive call makes two calls to itself on problems of half the
size, and six calls to matrix multiply over the semiring $(\min,+)$.
Here we analyze the algorithm in the \ourmodel.  The matrix multiplies
on two matrices of size $n \times n$ can be done in the model in
$Q_{\memsize}(n) = O(n^2(\wcost + n/\sqrt{\memsize}))$~\cite{BFGGS15}.
This leads to the recurrence
\[Q_{\smb{Kleene}}(n) = 2Q_{\smb{Kleene}}(n/2) + O(Q_{\memsize}(n)) + O(\wcost n^2)\]
which solves to $Q_{\smb{Kleene}}(n) = O(Q_{\memsize}(n))$ because the
cost is dominated at the root of the recurrence.  It is not known whether
this is optimal.  A similar approach
can be used for several other problems, including sequence alignment
with gaps, optimal binary search trees, and matrix chain
multiplication~\cite{CR06}.

Longest common subsequence (LCS) and edit distance (ED) are more
challenging, and covered next.

\section{Longest Common Subsequence and Edit Distance}
\label{sec:lcs}

This section describes a more efficient dynamic-programming algorithm
for longest common subsequence (LCS) and edit distance (ED).  The standard
approach for these problems (an $\memsize\times \memsize$ tiling)
results in an \iocost{} of $O(mn \wcost/\memsize)$ and \work{} of
$O(mn+mn\wcost/\memsize)$, where $m$ and $n$ are the length of the two
input strings.  Lemma~\ref{lem:diamond} states that the standard bound
is optimal under the standard DAG computation rule that all inputs
must be available before evaluating a node.  Perhaps surprisingly, we
are able to beat these bounds by leveraging the fact that dynamic
programs do not perform arbitrary functions at each node, and hence we
do not necessarily need all inputs to begin evaluating a node.

Our main result is captured by the following theorem for large input
strings.  For smaller strings, we can do even better (see Appendix~\ref{sec:LCS-app}).
\begin{theorem}\label{thm:LCSbound}
  Let $k_T=\min((\wcost{}/\memsize)^{1/3},\sqrt{\memsize})$ and
  suppose $m,n = \Omega(k_T\memsize)$.  Then it is possible to compute
  the ED or length of the LCS with \work{} $\W(m,n) =
  O(mn+mn\wcost/(k_T\memsize))$.

  Let $k_Q=\min(\wcost^{1/3},\sqrt{\memsize})$ and suppose
  $m,n=\Omega(k_Q\memsize)$.  Then it is possible to compute the ED or
  length of the LCS with an \iocost{} of
  $Q(m,n)=O(mn\wcost/(k_Q\memsize))$.
\end{theorem}
\noindent To understand these bounds, our algorithm beats the \iocost{} of the
standard tiling algorithm by a $k_Q$ factor.  And if
$\wcost \geq \memsize$, our algorithm (using different tuning
parameters) beats the \work{} of the standard tiling algorithm by a
$k_T$ factor.

\subparagraph*{Overview} The dynamic programs for LCS and ED correspond
to computing the shortest path through an $m \times n$ grid with
diagonal edges, where $m$ and $n$ are the string lengths.  We focus
here on computing the length of the shortest path, but it is possible
to output the path as well with the same asymptotic complexity
(see Appendix~\ref{sec:LCS-app}).  Without loss of generality, we
assume that $m \leq n$, so the grid is at least as wide as it is tall.
For LCS, all horizontal and vertical edges have weight 0; the diagonal
edges have weight $-1$ if the corresponding characters in the strings
match, and weight $\infty$ otherwise.  For ED, horizontal and vertical
edges have weight $1$, and diagonal edges have weights either $0$ or
$1$ depending on whether the characters match. Our algorithm
is not sensitive to the particular weights of the edges, and
thus it applies to both problems and their generalizations.

Note that the $m \times n$ grid is not built explicitly since building
and storing the graph would take $\Theta(mn)$ writes if $mn \gg
\memsize$.  To get any improvement, it is important that subgrids
reuse the same space.  The weights of each edge can be inferred by
reading the appropriate characters in each input string.

Our algorithm partitions the implicit grid into size-$(h\memsize'
\times k\memsize')$, where $h$ and $k$ are parameters of the algorithm
to be set later, and $\memsize' = \memsize/c$ for large enough
constant $c>1$ to give sufficient working space in \fastmem.  When
string lengths $m$ and $n \geq m$ are both ``large'', we use $h=k$ and
thus usually work with $k\memsize' \times k\memsize'$ square subgrids.
If the smaller string length $m$ is small enough, we instead use
parameters $h < k$.  To simplify the description of the algorithm, we
assume without loss of generality that $m$ and $n$ are divisible by
$h\memsize'$ and $k\memsize'$, respectively, and that $\memsize$ is
divisible by $c$.

Our algorithm operates on one $h\memsize' \times k\memsize'$ rectangle
at a time, where the edges are directed right and down. The
shortest-path distances to all nodes along the bottom and right
boundary of each rectangle are explicitly written out, but all other
intermediate computations are discarded.  We label the vertices
$u_{i,j}$ for $1 \leq i \leq h\memsize'$ and $1\leq j \leq k\memsize'$
according to their row and column in the square, respectively,
starting from the top-left corner.  We call the vertices immediately
above or to the left of the square the \emph{input nodes}.  The input
nodes are all outputs for some previously computed rectangle.  We call
the vertices $u_{h\memsize',j}$ along the bottom boundary and
$u_{i,k\memsize'}$ along the right boundary the \emph{output nodes}.

The goal is to reduce the number of writes, thereby decreasing the
overall cost of computing the output nodes, which we do by sacrificing
reads and \work{}.  It is not hard to see that recomputing internal
nodes enables us to reduce the number of writes.  Consider, for
example, the following simple approach assuming $M=\Theta(1)$: For
each output node of a $k\times k$ square, try all possible paths
through the square, keeping track of the best distance seen so far;
perform a write at the end to output the best value.\footnote{This
  approach requires constant \fastmem{} to keep the best distance, the
  current distance, and working space for computing the current
  distance.  We also need bits proportional to the path length to
  enumerate paths.}  Each output node tries $2^{\Theta(k)}$ paths, but
only a $\Theta(1/k)$-fraction of nodes are output nodes.  Setting
$k=\Theta(\lg \wcost)$ reduces the number of writes by a
$\Theta(\lg\wcost)$-factor at the cost of $\wcost^{O(1)}$ reads.  This
same approach can be extended to larger $M$, giving the same
$\lg \wcost$ improvement, by computing ``bands'' of nearby paths
simultaneously.  But our main algorithm, which we discuss next, is
much better as $\memsize$ gets larger (see
Theorem~\ref{thm:LCSbound}).

\subparagraph*{Path sketch} The key feature of the grid leveraged by our
algorithm is that shortest paths do not cross, which enables us to
avoid the exponential recomputation of the simple approach.  The
noncrossing property has been exploited previously for building
shortest-path data structures on the grid (e.g.,~\cite{Schmidt98}) and
more generally planar graphs
(e.g.,~\cite{FakcharoenpholRa01,KleinMoOr10}).  These previous
approaches do not consider the cost of writing to \slowmem{}, and they
build data structures that would be too large for our use.  Our
algorithm leverages the available \fastmem{} to compute bands of
nearby paths simultaneously.  We capture both the noncrossing and band
ideas through what we call a path sketch, which we define as follows.
The path sketch enables us to cheaply recompute the shortest paths to
nodes.

We call every $\memsize'$-th row in the square a \emph{superrow},
meaning there are $h$ superrows in the square.  The algorithm
partitions the $i$-th superrow into \emph{segments} $\langle
i,\ell,r\rangle$ of consecutive elements $u_{i\memsize',\ell},
u_{i\memsize',\ell+1},\ldots,u_{i\memsize',r}$.  The main restriction
on segments is that $r < \ell + \memsize'$, i.e., each segment
consists of at most $\memsize'$ consecutive elements in the superrow.
Note that the segment boundaries are determined by the algorithm and
are input dependent.

A \emph{path sketch} is a sequence of segments $\langle
s,\ell_s,r_s\rangle,\langle s+1, \ell_{s+1}, r_{s+1}\rangle, \langle
s+2,\ell_{s+2},r_{s+2}\rangle, \ldots,$ $\langle i,\ell_i,r_i\rangle$,
summarizing the shortest paths to the segment.  Specifically, this
sketch means that for each vertex in the last segment, there is a
shortest path to that vertex that goes through a vertex in each of the
segments in the sketch.  If the sketch starts at superrow $1$, then
the path originates from a node above the first superrow (i.e., the top
boundary or the topmost $\memsize'$ nodes of the left boundary).  If
the sketch starts with superrow $s>1$, then the path originates at one
of the $\memsize'$ nodes on the left boundary between superrows $s-1$
and $s$.  Since paths cannot go left, the path sketch also satisfies
$\ell_s \leq \ell_{s+1} \leq \cdots \leq \ell_i$.

\subparagraph*{Evaluating a path sketch}
Given a path sketch, we refer to the process of determining the
shortest-path distances to all nodes in the final segment
$\langle i,\ell_i,r_i\rangle$ as \emph{evaluating the path sketch} or
\emph{evaluating the segment}, with the distances in \fastmem{} when
the process completes.  Note that we have not yet described how to
build the path sketch, as the building process uses evaluation as a
subroutine.

\begin{figure}
\centering
\includegraphics[width=.55\textwidth]{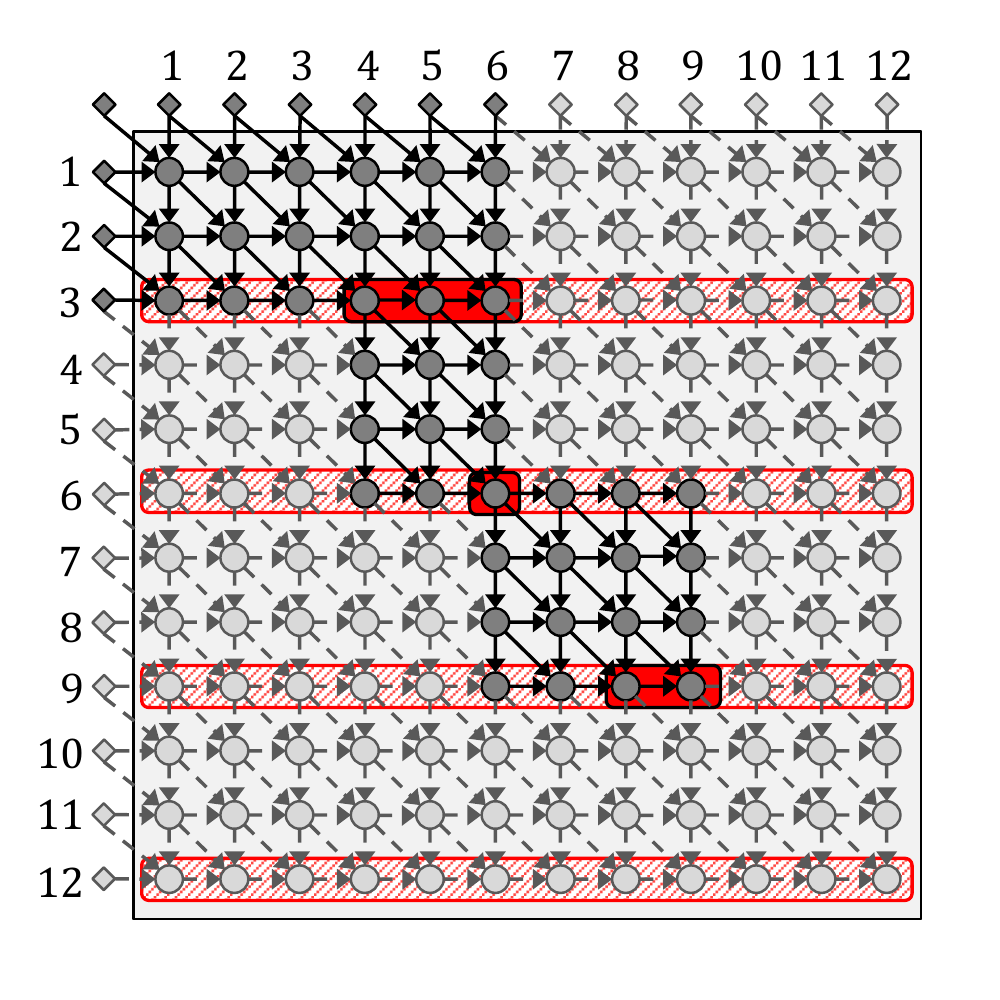}\vspace{-2em}
\caption{\small Example square grid and path sketch for
  $\memsize' = 3$ and $h=k=4$.  The circles are nodes in the square.
  The diamonds are input nodes (outputs of adjacent squares), omitting
  irrelevant edges.  The red slashes are the 4
  superrows, and the solid red are the sketch
  segments.}\label{fig:sketch}
\vspace{-0.1in}
\end{figure}

The main idea of evaluating the sketch is captured by
Figure~\ref{fig:sketch} for the example sketch
$\langle 1,4,6\rangle, \langle 2, 6, 6\rangle, \langle 3,8,9\rangle$.
The sketch tells us that shortest paths to $u_{9,8}$ and $u_{9,9}$
pass through one of $u_{3,4},u_{3,5},u_{3,6}$ and the node $u_{6,6}$.
Thus, to compute the distances to $u_{9,8}$ and $u_{9,9}$, we need
only consider paths through the darker nodes and solid edges---the
lighter nodes and dashed edges are not recomputed during evaluation.

The algorithm works as follows.  First compute the shortest-path
distances to the first segment in the sketch.  To do so, horizontally
sweep a height-$(\memsize'+1)$ column across the
$(\memsize'+1) \times k\memsize'$ slab raising above the $s$-th
superrow, keeping two columns in \fastmem{} at a time. Also keep the
newly computed distances to the first segment in \fastmem, and stop
the sweep at the right edge of the segment.  More generally, given the
distances to a segment in \fastmem, we can compute the values for the
next segment in the same manner by sweeping a column through the slab.
This algorithm yields the following performance.

\begin{lemma}\label{lem:usesketch}
  Given a path sketch $\langle s,\ell_s,r_s\rangle,\ldots,\langle
  i,\ell_i,r_i\rangle$ in an $h\memsize' \times k\memsize'$ grid with
  distances to all input nodes computed, our algorithm
  correctly computes the shortest-path distances to all nodes in the
  segment $\langle i,\ell_i,r_i\rangle$.  Assuming $k \geq h$ and
  \fastmem{} size $\memsize \geq 5\memsize' + \Theta(1)$, the
  algorithm requires $O(k\memsize^2)$ operations in \fastmem,
  $O(k\memsize)$ reads, and $0$ writes.
\end{lemma}
\begin{proof}
  Correctness follows from the definition of the path sketch: the
  sweep performed by the algorithm considers all possible paths that
  pass through these segments.

  The algorithm requires space in \fastmem{} to store two columns in
  the current slab, the previous segment in the sketch, and the next
  segment in the sketch, and the two segment boundaries themselves,
  totaling $4\memsize' + \Theta(1)$ \fastmem.  Due to the
  monotonically increasing left endpoints of each segment, the
  horizontal sweep repeats at most $\memsize'$ columns per supperrow,
  so the total number of column iterations is
  $O(k\memsize'+h\memsize') = O(k\memsize')$.  Multiplying by
  $\memsize'$ gives the number of nodes computed.

  The main contributor to reads is the input strings themselves to
  infer the structure/weights of the grid.  With $\memsize'$ additional
  \fastmem{}, we can store the ``vertical'' portion of the input
  string used while computing each slab, and thus the vertical string
  is read only once with $O(h\memsize')=O(k\memsize')$ reads.  The
  ``horizontal'' input characters can be read with each of the
  $O(k\memsize')$ column-sweep iterations.  An additional $k$ reads
  suffice to read the sketch itself, which is a lower-order term.
\end{proof}

\subparagraph*{Building the path sketch} The main algorithm on each
rectangle involves building the set of sketches to segments in the
bottom superrow.  At some point during the sketch-building process,
the distances to each output node is computed, at which point it can
be written out.  The main idea of the algorithm is a sketch-extension
subroutine: given segments in the $i$-th superrow and their sketches,
extend the sketches to produce segments in the $(i+1)$-th superrow
along with their sketches.

Our algorithm builds up an ordered list of consecutive path sketches,
one superrow at a time.  The first superrow is partitioned into $k$
segments, each containing exactly $\memsize'$ consecutive nodes.  The
list of sketches is initialized to these segments.

Given a list of sketches to the $i$-th superrow, our algorithm extends
the list of sketches to the $(i+1)$-th superrow as follows.  The
algorithm sweeps a height-$(\memsize'+1)$ column across the
$(\memsize'+1) \times k\memsize'$ slab between these superrows
(inclusive).  The sweep begins at the left end of the slab, reading
the input values from the left boundary, and continuing across the
entire width of the slab.  In \fastmem, we evaluate the first segment
of the $i$-th superrow (using the algorithm from
Lemma~\ref{lem:usesketch}).  Whenever the sweep crosses a segment
boundary in the $i$-th superrow, again evaluate the next segment in the
$i$-th superrow.  For each node in the slab, the sweep calculates both
the shortest-path distance and a pointer to the segment in the
previous superrow from whence the shortest path originates (or a null
pointer if it originates from the left boundary).  When the
originating segment of the bottom node (the node in the $(i+1)$-th
superrow) changes, the algorithm creates a new segment for the
$(i+1)$-th superrow and appends it to the sketch of the originating
segment.  If the segment in the current segment in the $(i+1)$-th
superrow grows past $\memsize'$ elements, a new segment is created
instead and the current path sketch is copied and spliced into the
list of sketches.
Any sketch that is not extended through this
process is no longer relevant and may be spliced out of the list of
sketches.  When the sweep reaches a node on the output boundary (right
edge or bottom edge of the square), the distance to that node is
written out.

\begin{lemma}\label{lem:segments}
  The sketching algorithm partitions the $i$-th superrow into at most
  $ik$ segments.
\end{lemma}
\begin{proof}
  The proof is by induction over superrows.  As a base case, the first
  superrow consists of exactly $k$ segments.  For the inductive step,
  there are two cases in which a new segment is started in the
  $(i+1)$-th superrow.  The first case is that the originating segment
  changes, which can occur at most $ik$ times by inductive assumption.
  The second case is that the current segment grows too large, which
  can occur at most $k$ times.  We thus have at most $(i+1)k$ segments
  in the $(i+1)$-th superrow.
\end{proof}

\begin{lemma}\label{lem:buildsketch}
  Suppose $h \leq k$ and \fastmem{} $\memsize \geq
  11\memsize'+\Theta(1)$, and consider an $h\memsize' \times
  k\memsize'$ grid with distances to input nodes already computed.
  Then the sketch building algorithm correctly computes the distances
  to all output boundary nodes using $O((hk)^2\memsize^2)$ operations
  in \fastmem, $O((hk)^2\memsize)$ reads from \slowmem, and
  $O(h^2k+X)$ writes to \slowmem, where $X=O(k\memsize)$ is the number
  of boundary nodes written out.
\end{lemma}
\begin{proof}
  Consider the cost of computing each slab, ignoring the writes to the
  output nodes.  We reserve $5\memsize'+\Theta(1)$ \fastmem{} for the
  process of evaluating segments in the previous superrow.  To perform
  the sweep in the current slab, we reserve $\memsize'$ \fastmem{} to
  store one segment in the previous row, $\memsize'$ \fastmem{} to
  store characters in the ``vertical'' input string, $4(\memsize'+1)$
  \fastmem{} to store two columns (each with distances and pointers)
  for the sweep, and an additional $\Theta(1)$ \fastmem{} to keep,
  e.g., the current segment boundaries.  Since there are at most $hk$
  segments in the previous superrow (Lemma~\ref{lem:segments}), the
  algorithm evaluates at most $hk$ segments; applying
  Lemma~\ref{lem:usesketch}, the cost is $O(hk^2\memsize^2)$
  operations, $O(hk^2\memsize)$ reads, and $0$ writes.  There are an
  additional $O(k\memsize^2)$ operations to sweep through the
  $k\memsize^2$ nodes in the slab, plus $O(k\memsize)$ reads to scan
  the ``horizontal'' input string.  Finally, there are $O(hk)$ writes
  to extend existing sketches and $O(hk)$ writes to copy at most $k$
  sketches.

  Summing across all $h$ slabs and accounting for the output nodes, we
  get $O((hk)^2\memsize^2 + hk\memsize^2)$ operations,
  $O((hk)^2\memsize+hk\memsize)$ reads, and $O(h^2k+X)$ writes.
  Removing the lower-order terms gives the lemma.
\end{proof}

Combining across all rectangles in the grid, we get the following
corollary.

\begin{corollary}\label{cor:LCScosts}
  Let $m\leq n$ be the length of the two input strings, with $m\geq
  \memsize$.  Suppose $h=O(m/\memsize)$ and $k=O(n/\memsize)$ with
  $h\leq k$.  Then it is possible to compute the LCS or edit distance
  of the strings with $O(mnhk)$ operations in \fastmem,
  $O(mnhk/\memsize)$ reads to \slowmem, and $O(mnh/\memsize^2 +
  mn/(h\memsize))$ writes to \slowmem.
\end{corollary}
\begin{proof}
  There are $\Theta(mn/(hk\memsize^2))$ size-$(h\memsize/11) \times
  (k\memsize/11)$ subgrids.  Multiplying by the cost of each grid
  (Lemma~\ref{lem:buildsketch}) gives
  the bound.
\end{proof}

Setting $h=k=1$ gives the standard $\memsize\times \memsize$ tiling
with $O(nm)$ \work{} and $O(mn\wcost/\memsize)$ \iocost{}.  As the
size of squares increase, the fraction of output nodes and hence
writes decreases, at the cost of more overhead for operations in
\fastmem{} and reads from \slowmem.  Assuming both $n$ and $m$ are
large enough to do so, plugging in $h=k=\max\{1,k_T\}$ or $h=k=k_Q$
with a few steps of algebra to eliminate terms yields
Theorem~\ref{thm:LCSbound}.

\section{Single-Source Shortest Paths}
\label{sec:dijkstra}

The single-source shortest-paths (SSSP) problem takes a directed
weighted graph $G=(V,E)$ and a source vertex $s\in V$, and outputs the
shortest distances $d(s,v)$ from $s$ to every other vertex in $v\in
V$.  For graphs with non-negative edge weights, the most efficient
algorithm is Dijkstra's algorithm~\cite{dijkstra1959}.

In this section we will study (variants of) Dijkstra's algorithm in
the asymmetric setting. We describe and analyze three versions (two
classical and one new variant) of Dijkstra's algorithm, and the best
version can be chosen based on the values of $\memsize$, $\wcost$, the
number of vertices $n=|V|$, and the number of edges $m=|E|$.

\begin{theorem}
\label{thm:Dijkstra}
The SSSP problem on a graph $G=(V,E)$ with non-negative edge weights
can be solved with $\displaystyle Q(n,m)=O\left(\min\left(n
    \left(\wcost + \frac{m}{\memsize}\right), \wcost(m+n\log
    n), m(\wcost+\log n)\right)\right)$ and $\W(n,m)=O(Q(n,m) + n\log n)$, both in expectation, on the \ourmodel.
\end{theorem}

We start with the classical Dijkstra's algorithm~\cite{dijkstra1959},
which maintains for each vertex $v$, $\delta(v)$, a tentative upper bound on the
distance, initialized to
$+\infty$ (except for $\delta(s)$, which is initialized to $0$).  The
algorithm consists of $n-1$ iterations, and the final distances from
$s$ are stored in $\delta(\cdot)$.  In each iteration, the
algorithm selects the unvisited vertex $u$ with smallest finite
$\delta(u)$, marks it as \emph{visited}, and uses its outgoing edges
to relax (update) all of its neighbors' distances.  A priority queue
is required to efficiently select the unvisited vertex with minimum
distance.  Using a Fibonacci heap~\cite{fredman1987fibonacci}, the
\work{} of the algorithm is $O(m+n\log n)$ in the standard (symmetric)
RAM model.  In the \ourmodel, the costs are $Q=\W=O((m+n\log
n)\wcost)$ since the Fibonacci heap requires asymptotically as many writes
as reads.  Alternatively, using a binary search tree for the priority
queue reduces the number of writes (see Section~\ref{sec:upper-intro})
at the cost of increasing the number of reads, giving $Q=\W=$ $O(m\log
n+\wcost{}m)$.  These bounds are better when $m=o(\wcost{}n)$.  Both
of these variants store the priority queue in \slowmem, requiring at
least one write to \slowmem{} per edge.

We now describe an algorithm, which we refer to as \emph{phased
  Dijkstra}, that fully maintains the priority queue in \fastmem{} and
only requires $O(n)$ writes to \slowmem{}.  The idea is to partition
the computation into phases such that for a parameter $\memsize'$ each
phase needs a priority queue of size at most $2\memsize'$ and visits
at least $\memsize'$ vertices.  By selecting $\memsize' = \memsize/c$
for an appropriate constant $c$, the priority queue fits in \fastmem, and the only
writes to \slowmem{} are the final distances.

Each phase starts and ends with an empty priority queue $P$ and
consists of two parts.  A Fibonacci heap is used for $P$, but is kept
small by discarding the $\memsize'$ largest elements (vertex distances)
whenever $|P| = 2\memsize'$.  To do this $P$ is flattened into an
array, the $\memsize'$-th smallest element $d_\smb{max}$ is found by
selection, and the Fibonacci heap is reconstructed from the elements
no greater than $d_\smb{max}$, all taking linear time.  All further
insertions in a given phase are not added to $P$ if they have a value
greater than $d_\smb{max}$.  The first part of each phase loops over all
edges in the graph and relaxes any that go from a visited to an
unvisited vertex (possibly inserting or decreasing a key in $P$).  The
second part then runs the standard Dijkstra's algorithm, repeatedly
visiting the vertex with minimum distance and relaxing its neighbors
until $P$ is empty.  To implement relax, the algorithm needs to know
whether a vertex is already in $P$, and if so its location in $P$ so that it
can do a decrease-key on it.  It is too costly to store this
information with the vertex in \slowmem{}, but it can be stored in
\fastmem{} using a hash table.

The correctness of this phased Dijkstra's algorithm follows from the fact that
it only ever visits the closest unvisited vertex, as with the standard
Dijkstra's algorithm.

\begin{lemma}
  Phased Dijkstra's has $\displaystyle Q(n,m)=O\left(n\left(\wcost + \frac{m}{\memsize}\right)\right)$ and
  $\W(n,m)=O(Q(n,m) +n\log n)$ both in expectation (for $M \leq n$).
\end{lemma}

\begin{proof}
During a phase either the size of $P$ will grow to $2\memsize'$ (and
hence delete some entries) or it will finish the algorithm.  If $P$
grows to $2\memsize'$ then at least $\memsize'$ vertices are visited
during the phase since that many need to be deleted with delete-min to
empty $P$.  Therefore the number of phases is at most $\lceil
n/\memsize'\rceil$.  Visiting all edges in the first part of each
phase involves at most $m$ insertions and decrease-keys into $P$, each
taking $O(1)$ amortized time in \fastmem, and $O(1)$ time to read the
edge from \slowmem.  Since compacting $Q$ when it overflows takes
linear time, its cost can be amortized against the insertions that
caused the overflow.  The cost across all phases for the first part is
therefore $Q = W = O(m \lceil n / \memsize' \rceil)$.  For the second
part, every vertex is visited once and every edge relaxed at most once
across all phases.  Visiting a vertex requires a delete-min in
\fastmem{} and a write to \slowmem, while relaxing an edge requires an
insert or decrease-key in \fastmem, and $O(1)$ reads from \slowmem{}.
We therefore have for this second part (across all phases) that $Q =
O(\wcost n + m)$ and $W = O(n(\wcost + \log n) + m)$.  The operations
on $P$ each include an expected $O(1)$ cost for the hash table
operations.  Together this gives our bounds.
\end{proof}
Compared to the first two versions of Dijkstra's algorithm with
$Q=\W=O(\wcost{}m+\min(\wcost{}n\log n,m\log n))$, the new algorithm
is strictly better when $\wcost \memsize>n$.
More specifically, the new algorithm performs better when
$nm/M<\max\{\wcost{}m,\min(\wcost{}n\log n,m\log n)\}$.
Combining these three algorithms proves Theorem~\ref{thm:Dijkstra},
when the best one is chosen based on the parameters $\memsize$,
$\wcost$, $n$, and $m$.


\section{Minimum Spanning Tree (MST)}
\label{sec:MST}

In this section we discuss several commonly-used algorithms for
computing a minimum spanning tree (MST) on a weighted graph $G=(V,E)$
with $n=|V|$ vertices and $m=|E|$ edges.  Some of them are optimal in
terms of the number of writes ($O(n)$).  Although loading a graph into
\slowmem{} requires $O(m)$ writes, the algorithms are still useful on
applications that compute multiple MSTs based on one input graph.  For
example, it can be useful for computing MSTs on subgraphs, such as
road maps, or when edge weights are time-varying functions and hence
the graph maintains its structure but the MST varies over time.

\myparagraph{Prim's algorithm} All three versions of Dijkstra's
algorithm discussed in Section~\ref{sec:dijkstra} can be adapted to
implement Prim's algorithm~\cite{CLRS}. Thus, the upper bounds of \iocost{} and
\work{} in Theorem~\ref{thm:Dijkstra} also hold for minimum spanning
trees.

\myparagraph{Kruskal's algorithm} The initial sorting phase requires
$Q=\W=O(m\log n+\wcost{}m)$~\cite{BFGGS15}.  The second phase
constructs a MST using union-find without path compression in $O(m\log
n)$ \work{}, and performs $O(n)$ writes (the actual edges of the
MST). Thus, the complexity is dominated by the first phase.  Neither
\iocost{} nor \work{} match our variant of \boruvka's algorithm.

\myparagraph{\boruvka's algorithm} \boruvka's algorithm consists of at
most $\log n$ rounds.  Initially all vertices belong in their own
component, and in each round, the lightest edges that connect each
component to another component are added to the edge set of the MST,
and components are merged using these edges.  This merging can be done
using, for example, depth first search among the components and hence
takes time proportional to the number of remaining components.
However, since edges are between original vertices the algorithm is
required to maintain a mapping from vertices to the component they
belong to.  Shortcutting all vertices on each round to point directly
to their component requires $O(n)$ writes per round and hence up to
$O(n \log n)$ total writes across the rounds.  This is not a
bottleneck in the standard RAM model but is in the asymmetric case.

We now describe a variant of \boruvka's algorithm which is
asymptotically optimal in the number of writes.  It requires only
$O(1)$ \fastmem.  The algorithm proceeds in two phases.  For the first
$\log \log n$ rounds, the algorithm performs no shortcuts (beyond the
merging of components).  Thus it will leave chains of length up to
$\log\log n$ that need to be followed to map each vertex to the
component it belongs to.  Since there are at most $O(m\log\log n)$
queries during the first $\log \log n$ rounds and each only require
reads, the total \work{} for identifying the minimum edges between
components in the first phase is $O(m(\log\log n)^2)$.  After the
first phase all vertices are shortcut to point to their component.  We
refer to these components as the phase-one components.  In the second
phase, on every round, we shortcut the phase-one components to point
directly to the component they belong to.  Since there can only be at
most $n / \log n$ phase-one components, and at most $\log n - \log
\log n$ rounds in phase-two, the total number of reads and writes for
these updates is $O(n)$.  During phase two the mapping from a vertex
to its component takes two steps: one to find its phase-one component
and another to get to the current component.  Therefore the total
\work{} for identifying the minimum edges between components in the
second phase is $O(m \log n)$.

All other work is on the components themselves (i.e. adding the forest
of minimum edges and performing DFS to merge components).  The number
of reads, writes, and other instructions is proportional to the number
of components. There are $n$ components on the first round and
the number decreases by at least a factor of two on each following round.
Therefore the total \work{} on the components is $O(\wcost n)$.
Summing the costs give the following lemma.

\begin{lemma}\label{lemma:boruvka} Our variant of
  \boruvka's algorithm generates a minimum spanning tree on a graph
  with $n$ vertices and $m$ edges with \iocost{} and \work{}
  $Q(n,m)=\W(n,m)=O(m\log n+\wcost{}n)$ on the \ourmodel.
\end{lemma}

\begin{theorem}
A minimum spanning tree on a graph $G=(V,E)$ can be computed with
\iocost{} $\displaystyle Q(n,m)=O\left(m \min\left( {n\over\memsize},
\log n\right) + \wcost n\right)$ and \work{} $\W(n,m)=O(Q(n,m) + n\log
n)$ on the \ourmodel.
\end{theorem}
The theorem is a combination of the bounds of Prim's and
\boruvka's algorithms (the $n/M$ term is in expectation).


\subsection*{Acknowledgments}
This research was supported in part by NSF grants CCF-1314590,
CCF-1314633 and CCF-1533858, the Intel Science and Technology Center
for Cloud Computing, and the Miller Institute for Basic Research in
Science at UC Berkeley.
\bibliographystyle{plain}
\bibliography{ref}

\begin{thebibliography}{10}

\bibitem{AggarwalV88}
Alok Aggarwal and Jeffrey~S. Vitter.
\newblock The {I}nput/{O}utput complexity of sorting and related problems.
\newblock {\em Communications of the ACM}, 31(9), 1988.

\bibitem{Aho74}
Alfred~V. Aho, John~E. Hopcroft, and Jeffrey~D. Ullman.
\newblock {\em The Design and Analysis of Computer Algorithms}.
\newblock Addison-Wesley, Reading, MA, 1974.

\bibitem{AKS83}
Mikl\'{o}s Ajtai, J\'{a}nos Koml\'{o}s, and Endre Szemer\'{e}di.
\newblock An $\mb{O}(n \log n)$ sorting network.
\newblock In {\em Proc.~ACM Symposium on Theory of Computing (STOC)}, 1983.

\bibitem{Ajwani2009}
Deepak Ajwani, Andreas Beckmann, Riko Jacob, Ulrich Meyer, and Gabriel Moruz.
\newblock On computational models for flash memory devices.
\newblock In {\em Proc.~ACM International Symposium on Experimental Algorithms
  (SEA)}, 2009.

\bibitem{Akel11}
Ameen Akel, Adrian~M. Caulfield, Todor~I. Mollov, Rajech~K. Gupta, and Steven
  Swanson.
\newblock Onyx: A prototype phase change memory storage array.
\newblock In {\em Proc.~USENIX Workshop on Hot Topics in Storage and File
  Systems (HotStorage)}, 2011.

\bibitem{Athanassoulis12}
Manos Athanassoulis, Bishwaranjan Bhattacharjee, Mustafa Canim, and Kenneth~A.
  Ross.
\newblock Path processing using solid state storage.
\newblock In {\em Proc.~International Workshop on Accelerating Data Management
  Systems Using Modern Processor and Storage Architectures (ADMS)}, 2012.

\bibitem{BT06}
Avraham Ben-Aroya and Sivan Toledo.
\newblock Competitive analysis of flash-memory algorithms.
\newblock In {\em Proc.~European Symposium on Algorithms (ESA)}, 2006.

\bibitem{BBFGGMS16}
Naama Ben-David, Guy~E. Blelloch, Jeremy~T. Fineman, Phillip~B. Gibbons, Yan
  Gu, Charlie McGuffey, and Julian Shun.
\newblock Parallel algorithms for asymmetric read-write costs.
\newblock In {\em Proc. ACM Symposium on Parallelism in Algorithms and
  Architectures (SPAA)}, 2016.

\bibitem{BFGGS15}
Guy~E. Blelloch, Jeremy~T. Fineman, Phillip~B. Gibbons, Yan Gu, and Julian
  Shun.
\newblock Sorting with asymmetric read and write costs.
\newblock In {\em Proc. ACM Symposium on Parallelism in Algorithms and
  Architectures (SPAA)}, 2015.

\bibitem{Carsonetal15}
Erin Carson, James Demmel, Laura Grigori, Nicholas Knight, Penporn
  Koanantakool, Oded Schwartz, and Harsha~V. Simhadri.
\newblock Write-avoiding algorithms.
\newblock In {\em Proc.~IEEE International Parallel \& Distributed Processing
  Symposium (IPDPS)}, 2016.

\bibitem{Chen11}
Shimin Chen, Phillip~B. Gibbons, and Suman Nath.
\newblock Rethinking database algorithms for phase change memory.
\newblock In {\em Proc.~Conference on Innovative Data Systems Research (CIDR)},
  2011.

\bibitem{Chen15}
Shimin Chen and Qin Jin.
\newblock Persistent {B$^+$}-trees in non-volatile main memory.
\newblock {\em PVLDB}, 8(7), 2015.

\bibitem{ChoL09}
Sangyeun Cho and Hyunjin Lee.
\newblock {Flip-N-Write}: A simple deterministic technique to improve {PRAM}
  write performance, energy and endurance.
\newblock In {\em Proc.~IEEE/ACM International Symposium on Microarchitecture
  (MICRO)}, 2009.

\bibitem{CR06}
Rezaul~A. Chowdhury and Vijaya Ramachandran.
\newblock Cache-oblivious dynamic programming.
\newblock In {\em Proc. {ACM-SIAM} Symposium on Discrete Algorithms ({SODA})},
  2006.

\bibitem{CS76}
Stephen Cook and Ravi Sethi.
\newblock Storage requirements for deterministic polynomial time recognizable
  languages.
\newblock {\em Journal of Computer and System Sciences}, 13(1), 1976.

\bibitem{CLRS}
Thomas~H. Cormen, Charles~E. Leiserson, Ronald~L. Rivest, and Clifford Stein.
\newblock {\em Introduction to Algorithms (3rd edition)}.
\newblock MIT Press, 2009.

\bibitem{DCKO08}
Mark de~Berg, Otfried Cheong, Mark van Kreveld, and Mark Overmars.
\newblock {\em Computational Geometry: Algorithms and Applications}.
\newblock Springer-Verlag, 2008.

\bibitem{dijkstra1959}
Edsger~W. Dijkstra.
\newblock A note on two problems in connexion with graphs.
\newblock {\em Numerische mathematik}, 1(1), 1959.

\bibitem{Dong09}
Xiangyu Dong, Norman~P. Jouupi, and Yuan Xie.
\newblock {PCRAMsim}: System-level performance, energy, and area modeling for
  phase-change {RAM}.
\newblock In {\em Proc.~ACM International Conference on Computer-Aided Design
  (ICCAD)}, 2009.

\bibitem{Dong08}
Xiangyu Dong, Xiaoxia Wu, Guangyu Sun, Yuan Xie, Hai~H. Li, and Yiran Chen.
\newblock Circuit and microarchitecture evaluation of {3D} stacking magnetic
  {RAM (MRAM)} as a universal memory replacement.
\newblock In {\em Proc.~ACM Design Automation Conference (DAC)}, 2008.

\bibitem{Eppstein14}
David Eppstein, Michael~T. Goodrich, Michael Mitzenmacher, and Pawel Pszona.
\newblock Wear minimization for cuckoo hashing: How not to throw a lot of eggs
  into one basket.
\newblock In {\em Proc.~ACM International Symposium on Experimental Algorithms
  (SEA)}, 2014.

\bibitem{FakcharoenpholRa01}
Jittat Fakcharoenphol and Satish Rao.
\newblock Planar graphs, negative weight edges, shortest paths, near linear
  time.
\newblock In {\em Proc.~IEEE Symposium on Foundations of Computer Science
  (FOCS)}, 2001.

\bibitem{Floyd:1962}
Robert~W. Floyd.
\newblock Algorithm 97: Shortest path.
\newblock {\em Commun. ACM}, 5(6):345--, June 1962.

\bibitem{fredman1987fibonacci}
Michael~L. Fredman and Robert~E. Tarjan.
\newblock Fibonacci heaps and their uses in improved network optimization
  algorithms.
\newblock {\em Journal of the ACM}, 34(3), 1987.

\bibitem{Gal05}
Eran Gal and Sivan Toledo.
\newblock Algorithms and data structures for flash memories.
\newblock {\em ACM Computing Surveys}, 37(2), 2005.

\bibitem{Goodrich14}
Michael~T. Goodrich.
\newblock Zig-zag sort: A simple deterministic data-oblivious sorting algorithm
  running in ${O}(n \log n)$ time.
\newblock In {\em Proc. ACM Symposium on Theory of Computing (STOC)}, 2014.

\bibitem{Hirschberg75}
D.~S. Hirschberg.
\newblock A linear space algorithm for computing maximal common subsequences.
\newblock {\em Commun. ACM}, 18(6):341--343, June 1975.

\bibitem{HK81}
Jia{-}Wei Hong and H.~T. Kung.
\newblock {I/O} complexity: The red-blue pebble game.
\newblock In {\em Proc. {ACM} Symposium on Theory of Computing (STOC)}, 1981.

\bibitem{HPV77}
John Hopcroft, Wolfgang Paul, and Leslie Valiant.
\newblock On time versus space.
\newblock {\em Journal of the ACM}, 24(2), 1977.

\bibitem{hp-nvm15}
{HP, SanDisk} partner on memristor, {ReRAM} technology.
\newblock http://www.bit-tech.net/news/
  hardware/2015/10/09/hp-sandisk-reram-memristor, October 2015.

\bibitem{HuZXTGS14}
Jingtong Hu, Qingfeng Zhuge, Chun~Jason Xue, Wei-Che Tseng, Shouzhen Gu, and
  Edwin Sha.
\newblock Scheduling to optimize cache utilization for non-volatile main
  memories.
\newblock {\em IEEE Transactions on Computers}, 63(8), 2014.

\bibitem{ibm-pcm14b}
www.slideshare.net/IBMZRL/theseus-pss-nvmw2014, 2014.

\bibitem{intel-nvm15}
Intel and {M}icron produce breakthrough memory technology.
\newblock http://newsroom.intel.com/
  community/intel\_newsroom/blog/2015/07/28/intel-and-micron-produce-breakthrough-memory-technology,
  July 2015.

\bibitem{Kim14}
Hyojun Kim, Sangeetha Seshadri, Clement~L. Dickey, and Lawrence Chu.
\newblock Evaluating phase change memory for enterprise storage systems: A
  study of caching and tiering approaches.
\newblock In {\em Proc.~USENIX Conference on File and Storage Technologies
  (FAST)}, 2014.

\bibitem{KleinMoOr10}
Philip~N. Klein, Shay Mozes, and Oren Weimann.
\newblock Shortest paths in directed planar graphs with negative lengths: A
  linear-space ${O}(n \log^2 n)$-time algorithm.
\newblock {\em ACM Transactions on Algorithms}, 6(2), 2010.

\bibitem{LeeIMB09}
Benjamin~C. Lee, Engin Ipek, Onur Mutlu, and Doug Burger.
\newblock Architecting phase change memory as a scalable {DRAM} alternative.
\newblock In {\em Proc.~ACM International Symposium on Computer Architecture
  (ISCA)}, 2009.

\bibitem{Meena14}
Jagan~S. Meena, Simon~M. Sze, Umesh Chand, and Tseung-Yuan Tseng.
\newblock Overview of emerging nonvolatile memory technologies.
\newblock {\em Nanoscale Research Letters}, 9, 2014.

\bibitem{nath:vldbj10}
Suman Nath and Phillip~B. Gibbons.
\newblock Online maintenance of very large random samples on flash storage.
\newblock {\em VLDB J.}, 19(1), 2010.

\bibitem{ParkS09}
Hyoungmin Park and Kyuseok Shim.
\newblock {FAST}: Flash-aware external sorting for mobile database systems.
\newblock {\em Journal of Systems and Software}, 82(8), 2009.

\bibitem{PPP04}
Joon{-}Sang Park, Michael Penner, and Viktor~K. Prasanna.
\newblock Optimizing graph algorithms for improved cache performance.
\newblock {\em {IEEE} Transactions on Parallel and Distributed Systems}, 15(9),
  2004.

\bibitem{Paterson90}
Mike~S. Paterson.
\newblock Improved sorting networks with $\mb{O}(\log n)$ depth.
\newblock {\em Algorithmica}, 5(1), 1990.

\bibitem{Qureshi12}
Moinuddin~K. Qureshi, Sudhanva Gurumurthi, and Bipin Rajendran.
\newblock {\em Phase Change Memory: From Devices to Systems}.
\newblock Morgan \& Claypool, 2011.

\bibitem{Schmidt98}
Jeanette~P. Schmidt.
\newblock All shortest paths in weighted grid graphs and its application to
  finding all approximate repeats in strings.
\newblock {\em SIAM Journal on Computing}, 27, 1998.

\bibitem{Seiferas09}
Joel Seiferas.
\newblock Sorting networks of logarithmic depth, further simplified.
\newblock {\em Algorithmica}, 53(3), 2009.

\bibitem{Tarjan83}
Robert~E. Tarjan.
\newblock Updating a balanced search tree in ${O}(1)$ rotations.
\newblock {\em Information Processing Letters}, 16(5), 1983.

\bibitem{Viglas12}
Stratis~D. Viglas.
\newblock Adapting the {B}$^+$-tree for asymmetric {I/O}.
\newblock In {\em Proc.~East European Conference on Advances in Databases and
  Information Systems (ADBIS)}, 2012.

\bibitem{Viglas14}
Stratis~D. Viglas.
\newblock Write-limited sorts and joins for persistent memory.
\newblock {\em {PVLDB}}, 7(5), 2014.

\bibitem{Xu11}
Cong Xu, Xiangyu Dong, Norman~P. Jouppi, and Yuan Xie.
\newblock Design implications of memristor-based {RRAM} cross-point structures.
\newblock In {\em Proc.~IEEE Design, Automation and Test in Europe (DATE)},
  2011.

\bibitem{yang:iscas07}
Byung-Do Yang, Jae-Eun Lee, Jang-Su Kim, Junghyun Cho, Seung-Yun Lee, and
  Byoung-Gon Yu.
\newblock A low power phase-change random access memory using a data-comparison
  write scheme.
\newblock In {\em Proc.~IEEE International Symposium on Circuits and Systems
  (ISCAS)}, 2007.

\bibitem{Yole13}
{Yole Developpement}.
\newblock Emerging non-volatile memory technologies, 2013.

\bibitem{ZhouZYZ09}
Ping Zhou, Bo~Zhao, Jun Yang, and Youtao Zhang.
\newblock A durable and energy efficient main memory using phase change memory
  technology.
\newblock In {\em Proc.~ACM International Symposium on Computer Architecture
  (ISCA)}, 2009.

\bibitem{ZWT13}
Omer Zilberberg, Shlomo Weiss, and Sivan Toledo.
\newblock Phase-change memory: An architectural perspective.
\newblock {\em ACM Computing Surveys}, 45(3), 2013.

\end{thebibliography}

\appendix

\section{Longest Common Subsequence: Further Results}\label{sec:LCS-app}

\newenvironment{reflemma}[1]{\begin{trivlist}\item[\hskip
\labelsep{\bf Lemma \ref{lem:#1}}]\it}
{\end{trivlist}}

\newenvironment{reftheorem}[1]{\begin{trivlist}\item[\hskip
\labelsep{\bf Theorem \ref{thm:#1}}]\it}
{\end{trivlist}}

\begin{reftheorem}{LCSbound}
  Let $k_T=\min\{(\wcost{}/\memsize)^{1/3},\sqrt{\memsize}\}$ and
  suppose $m,n = \Omega(k_T\memsize)$.  Then it is possible to compute
  the length of the LCS or edit distance with total \work{} $\W(m,n) =
  O(mn+mn\wcost/(\memsize k_T))$.

  Let $k_Q=\min\{\wcost^{1/3},\sqrt{\memsize}\}$ and suppose
  $m,n=\Omega(k_Q\memsize)$.  Then it is possible to compute the
  length of the LCS or edit distance with an \iocost{} of
  $Q(m,n)=O(mn\wcost/(k_Q\memsize))$.
\end{reftheorem}
\begin{proof}
  As long as $h \leq \sqrt{\memsize}$, the number of writes reduces to
  $O(mn/(h\memsize))$.  (Increasing $h$ further causes the number of
  writes to increase.)

  Consider the \work{} bound first.  If $k_T \leq 1$, then just use
  algorithm with $h=k=1$.  Otherwise, let $\memsize' = \memsize/11$
  and use the algorithm with $h = k = k_T$. As long as $h=k \leq
  (\wcost/\memsize)^{1/3}$, which is true for $k_T$, the \work{} of
  operations is less than the \work{} of writes, giving the bound.

  For the \iocost{}, use our algorithm with $h=k=k_Q$.  As long as
  $h=k \leq \wcost^{1/3}$, then cost of reads is less than the cost of
  writes.
\end{proof}

\myparagraph{Improving the bound for smaller string lengths}
If $m \leq \memsize$, then the standard I/O algorithm becomes even
better --- simply sweep a column through, which remains in \fastmem,
using $O(m+n)$ reads, no writes, and $O(mn)$ \work.  Since there are no
writes, we cannot beat that bound.  As described already, if $m \geq
k\memsize'$ then our algorithm partitions the grid into $k\memsize'
\times k\memsize'$ squares, which for larger $k$ saves writes by
sacrificing reads and re-computation.  The remaining question is what
happens when $m$ is larger than \fastmem{} but not too much larger,
i.e., $\memsize < m < k_T\memsize'$ or $\memsize < m < k_Q\memsize'$.

When $m$ falls in this range, we apply the algorithm to $m \times
k\memsize'$ rectangles, i.e., setting $h = m/\memsize'$.  It turns out
we can achieve a better bound than Theorem~\ref{thm:LCSbound} by
increasing $k$ even further.  The key observation here is that the
bottom of the rectangle no longer needs to be written out because
there is no rectangle below it --- only the right edge is an output
edge.  The number of writes per rectangle (Lemma~\ref{lem:buildsketch}
with $X=h\memsize'$) reduces to $O(h^2k + h\memsize)$.  We thus have
the following modified version of Corollary~\ref{cor:LCScosts}

\begin{corollary}\label{cor:shortLCScosts}
  Let $m\leq n$ be the length of the two input strings, with $m\geq
  \memsize$.  Let $h=\Theta(m/\memsize)$, and suppose
  $k=O(n/\memsize)$ satisfying $h\leq k$. Then it is possible to
  compute LCS or edit distance of a length $m$ and $n$ input strings
  with $O(mnhk)$ operations in \fastmem, $O(mnhk/\memsize)$ reads to
  \slowmem, and $O(mnh/\memsize^2 + mn/(k\memsize))$ writes to
  \slowmem.
\end{corollary}
\begin{proof}
  There are $\Theta(n/(k\memsize))$ size $m \times (k\memsize/11)$
  subgrids.  Multiplying by the cost of each grid from
  Lemma~\ref{lem:buildsketch}, with $X=h\memsize$, gives
  $O(nh^2k\memsize)$ operations, $O(nh^2k)$ reads, and
  $O(nh^2/\memsize + nh/k)$ writes.  Substituting one of the $h$ terms
  with $h=\Theta(m/\memsize)$ gives the theorem.
\end{proof}

The following theorem provides the improved \work{} and \iocost{} in the
case that one string is short but the other is long.  To understand
the bounds here, consider the maximum and minimum values of $h$ for
$h=m/\memsize'$ and large $\wcost$.  If $h=(\wcost/\memsize)^{1/3}$,
i.e., $m$ is large enough that we can divide into $k_T\memsize' \times
k_T\memsize'$ squares, then we get $k_T'= (\wcost/\memsize)^{1/3}$
matching the bound in Theorem~\ref{thm:LCSbound}.  As $h$ decreases,
the bound improves. In the limit, $h=\Theta(1)$ (or
$m=\Theta(\memsize)$), we get $k_T' = \sqrt{\wcost/\memsize}$ which is
better.

\begin{theorem}
  Let $h=\Theta(m/\memsize)$ and suppose that $h \leq k_T$ specified
  in Theorem~\ref{thm:LCSbound}.  Let $k_T' =
  \min\{\sqrt{\wcost/(h\memsize)},\memsize/h\}$ and suppose that
  $n=\Omega(k_T'\memsize)$.  Then it is possible to compute the length
  of the LCS or edit distance with total \work{} of $\W(m,n) =
  O(mn+mn\wcost/(k_T'\memsize))$.

  Let $h=\Theta(m/\memsize)$ and suppose that $h \leq k_Q$ specified
  in Theorem~\ref{thm:LCSbound}.  Let $k_Q' =
  \min\{\sqrt{\wcost/h},\memsize/h\}$ and suppose
  $n=\Omega(k_Q'\memsize)$.  Then it is possible to compute the length
  of the LCS or edit distance with an \iocost{} of
  $Q(m,n)=O(mn\wcost/(k_Q'\memsize))$.
\end{theorem}
\begin{proof}
  With the restrictions on $h$, we have $h\leq k$, so
  Corollary~\ref{cor:shortLCScosts} is applicable.  As in proof of
  Theorem~\ref{thm:LCSbound}, the second term of the min has the
  effect that $O(mnh/\memsize^2 + mn/(k\memsize)) =
  O(mn/(k\memsize))$.  The rest of the bound follows by choice of $k$
  to makes the cost of writes dominate.
\end{proof}

When $n$ is also small, the bound improves further.  In this case, the
algorithm consists of building the sketch on a single $m \times n$
grid, so no boundary nodes are output --- the only writes that need be
performed are the sketch itself.

\begin{theorem}
  Let $m\leq n$ be the length of the two input strings, with $m\geq
  \memsize$.  Let $h=\Theta(m/\memsize)$ and let
  $k=\Theta(n/\memsize)$.  Then it is possible to compute the LCS or
  edit distance of the two strings with \work{} $\W(m,n) = O(mnhk +
  h^2k\wcost)$ and \iocost{} $Q(m,n) = O(mnhk/\memsize + h^2k\wcost)$.
\end{theorem}
\begin{proof}
  The bound follows directly from Lemma~\ref{lem:buildsketch} with
  $X=0$ and substituting one $hk$ term in the read/\work{} bounds.
\end{proof}

\myparagraph{Recovering the shortest path} The standard approach for
outputting the shortest path is to trace backwards through the grid
from the bottom-rightmost node.  This approach assumes that the
distances to all internal nodes are known, but unfortunately our
algorithm discards distances to interior nodes.

Fortunately, the sketch provides enough information to cheaply
traceback a path through each square without any additional writes
(except the path itself) and without asymptotically more reads or
\work{}.  In particular, for any node $v$ in superrow $i+1$, it is not
hard to identify a node $u$ in superrow $i$ such that a shortest path
to $v$ passes through $u$.  Consider the sketch
$\langle s,\ell_s,r_s \rangle, \ldots, \langle i,\ell_i,r_i\rangle,
\langle i+1, \ell_{i+1}, r_{i+1}\rangle$
to the segment $\langle i+1, \ell_{i+1}, r_{i+1}\rangle$ that includes
$v$.  The vertex $u$ is one of the vertices in the penultimate segment
$\langle i, \ell_i, r_i \rangle$ of the sketch, so the goal is to
identify which one.  To do so, evaluate the sketch to the segment
$\langle i, \ell_i, r_i\rangle$.  Then perform a horizontal sweep
through the final slab, keeping track of the originating vertex from
the penultimate segment.

Now suppose we have these vertices $u$ and $v$ that fall along the
shortest path and are in consecutive superrows.  We also need to
identify the path through the slab between $u$ and $v$.  To do that,
we apply Hirschberg's~\cite{Hirschberg75} recursive low-space
algorithm for path recovery in the ED/LCS grid, splitting the
horizontal dimension in half on each recursion.  Note that the \work{}
reduces by a constant fraction in each recursion, but the \iocost{}
does not (the only \iocost{} here is from reading the ``horizontal''
input string), so it may not be immediately obvious that the \iocost{}
is cheap enough.  Fortunately, after recursing at most $\lg k$ times,
the length of the horizontal substring is at most $M'$ and the
remaining path-recovery subproblem can be done with no further reads
from \slowmem{}.

Putting it all together, tracing a path to the previous superrow
requires one sketch evaluation, followed by \work{} that is linear in
the area and an \iocost{} that corresponds to reading the
horizontal string from \slowmem{} $\lg k$ times.
Rounding up loosely, we get \work{} of $O(k(\memsize')^2 +
(\memsize)(k\memsize)) = O(k\memsize^2)$ along with $O(k\memsize' +
(k\memsize')\lg k) = O(k^2\memsize)$ reads.  Multiplying by the $h$
superrows, we have $O(hk\memsize^2)$ \work{} and $O(hk^2\memsize)$
reads from \slowmem{}.  Both of these are less than the cost of
building the sketch in the first place (Lemma~\ref{lem:buildsketch}).

\section{Write-Efficient Floyd--Warshall Algorithm}
\label{sec:fw}

The Floyd--Warshall algorithm solves the all-pairs shortest path
problem on weighted graphs~\cite{Floyd:1962}.  It requires $O(n^3)$
writes in its original and currently recognized form and hence
$O(\wcost n^3)$ \work{} in \ourmodel{}.  Here we describe how to
reorganize the computation so that it only requires
$\W=O(n^3+\wcost{}n^2)$ \work. Compared to the version described in
Section~\ref{sec:upper-intro}, this algorithm is easier to program,
and has a smaller constant coefficient.

Consider a graph $G$ with vertices $V = \{1,\ldots,n\}$ and a
function $\SP(i, j, k)$ that returns the shortest possible path from
$i$ to $j$ using vertices only from the set $\{1,2,\ldots,k\}$ as
intermediate points along the path.  $\SP(i, j, k)$ can be computed
using the Floyd--Warshall algorithm with the following update rule:
$$\SP(i,j,k)=\min\{\SP(i,j,k-1),\SP(i,k,k-1)+\SP(k,j,k-1)\}$$ where
$\SP(i,j,0)$ is initialized as the weight of the edge from vertex $i$
to vertex $j$ ($+\infty$ if the edge does not exist).  The shortest
path from vertex $i$ to vertex $j$ is stored in $\SP(i,j,n)$, after
the computation is finished.

With this formulation of the DP, $O(n^3)$ writes will be required.  We now introduce an
alternative and equivalent formulation of the DP that will reduce the number of writes
to $O(n^2)$.  Let $A(i,k)$ be $\SP(i,k,k-1)$ and $B(k,j)$ be
$\SP(k,j,k-1)$ (with the same initialization as before for $A(i,1)$ and $B(1,i)$).  Then they can be
computed as:
$$A(i,k) = \min_{k'<k} \{A(i,k')+B(k',k)\}$$
$$B(k,j) = \min_{k'<k} \{A(k,k')+B(k',j)\}$$ in increasing order of
$k$ from $1$ to $n$.  After $A(i,k)$ and $B(k,j)$ are calculated, the
shortest distance $D(i,j)$ from vertex $i$ to vertex $j$ (equivalent to
$\SP(i,j,n)$) can be computed as:
$$D(i,j) = \min_{1\le k \le n} \{A(i,k)+B(k,j)\}$$

Clearly the computation of the new DP requires only $O(n^2)$ writes.
The correctness can be easily shown by verifying that $D(i,j)$ is equivalent
to $\SP(i,j,n)$.

\end{document}